\documentclass[a4paper,UKenglish,cleveref,authorcolumns,dvipsnames]{lipics-v2021}

\title{Communication Complexity of Collision}

\author{Mika G\"o\"os}{EPFL, Switzerland}{mika.goos@epfl.ch}{}{}

\author{Siddhartha Jain}{EPFL, Switzerland}{siddhartha.jain@epfl.ch}{https://orcid.org/0000-0003-2142-5801}{}

\authorrunning{M. G\"o\"os, S. Jain}
\Copyright{Mika G\"o\"os, Siddhartha Jain}

\ccsdesc[500]{Theory of computation~Communication complexity}

\keywords{Collision, Communication complexity, Lifting}

\acknowledgements{We thank anonymous RANDOM reviewers for their helpful comments.}

\nolinenumbers 

\EventEditors{Amit Chakrabarti and Chaitanya Swamy}
\EventNoEds{2}
\EventLongTitle{Approximation, Randomization, and Combinatorial Optimization. Algorithms and Techniques (APPROX/RANDOM 2022)}
\EventShortTitle{\mbox{\scriptsize APPROX/RANDOM 2022}}
\EventAcronym{APPROX/RANDOM}
\EventYear{2022}
\EventDate{September 19--21, 2022}
\EventLocation{University of Illinois, Urbana-Champaign, USA (Virtual Conference)}
\EventLogo{}
\SeriesVolume{245}
\ArticleNo{19}

\category{RANDOM}
\relatedversion{}\relatedversiondetails{Full Version}{https://eccc.weizmann.ac.il/report/2022/096/}

\usepackage{amsfonts,amssymb,amsmath,amsthm,mathtools,thmtools}
\usepackage{microtype,xspace,bm,floatrow}

\usepackage{doi}
\usepackage{hyperref}
\usepackage[dvipsnames]{xcolor}
\hypersetup{
	colorlinks=true,
	urlcolor=MidnightBlue,
	linkcolor=MidnightBlue,
	citecolor=MidnightBlue,
	unicode
}

\usepackage{bold-extra}

\usepackage{enumitem}
\setlist{  
  listparindent=\parindent,
  parsep=0pt,
}

\crefname{equation}{}{}

\usepackage{tikz}
\usepackage{graphicx}
\usetikzlibrary{positioning,arrows.meta,math,shapes.geometric,decorations.pathmorphing}

\newlength{\nodeline}
\newlength{\arrowline}
\definecolor{color_gadget_PHP}{RGB}{219, 48, 122}
\colorlet{color_gadget_PHP_inner}{color_gadget_PHP!10!white}
\colorlet{color_gadget_PHP_label}{color_gadget_PHP!70!black}

\definecolor{color_gadget_SOPL}{RGB}{255, 153, 0}
\colorlet{color_gadget_SOPL_inner}{color_gadget_SOPL!10!white}
\colorlet{color_gadget_SOPL_label}{color_gadget_SOPL!70!black}

\definecolor{color_gadget_PATHPHP}{RGB}{0, 153, 255}
\colorlet{color_gadget_PATHPHP_inner}{color_gadget_PATHPHP!10!white}
\colorlet{color_gadget_PATHPHP_label}{color_gadget_PATHPHP!70!black}

\definecolor{color_gadget_ITER}{RGB}{153, 204, 0}
\colorlet{color_gadget_ITER_inner}{color_gadget_ITER!10!white}
\colorlet{color_gadget_ITER_label}{color_gadget_ITER!70!black}

\definecolor{color_gadget_EOPL}{RGB}{0, 153, 255}
\colorlet{color_gadget_EOPL_inner}{color_gadget_EOPL!10!white}
\colorlet{color_gadget_EOPL_label}{color_gadget_EOPL!70!black}

\tikzstyle{node} = [rectangle, line width=\nodeline, draw = black, fill = white, inner sep = 0mm, minimum size = 3.5mm]
\tikzstyle{node_small} = [node, circle, line width = 0.25mm, minimum size = 2.5mm]
\tikzstyle{solution} = [fill=red!90!,draw=black!50!red]
\tikzstyle{side} = [fill=Goldenrod,draw=Brown]
\tikzstyle{node_text} = [] 

\tikzstyle{node_regular} = [node]
\tikzstyle{node_regular_small} = [node_small]

\tikzstyle{node_solution} = [node, solution]
\tikzstyle{node_solution_small} = [node_small,solution]

\tikzstyle{node_a} = [node, side, rectangle, minimum size = 4.3mm]
\tikzstyle{node_a_solution} = [node_a, solution]
\tikzstyle{node_a_small} = [node_small, side, rectangle, minimum size = 2.15mm]
\tikzstyle{node_a_solution_small} = [node_a_small,solution]

\tikzstyle{node_b} = [node, side, diamond, minimum size = 6.2mm]
\tikzstyle{node_b_small} = [node_a_small, diamond, minimum size = 3.1mm]

\tikzstyle{node_notice} = [node, draw = Green!20!LimeGreen, line width=2.5\nodeline, fill=none, minimum size = 9mm, dotted]
\tikzstyle{node_notice_small} = [node_notice, line width=2*\nodeline, minimum size = 8mm]

\tikzstyle{naive} = [minimum size = 4.5mm,rectangle]
\tikzstyle{node_regular_intro} = [node,fill=Gray!10!white]

\tikzstyle{edge} = [-{Latex[round]}, line width=\arrowline]

\tikzstyle{edge_regular} = [edge]
\tikzstyle{edge_regular_small} = [-{Latex[round]}, line width = 0.25mm, shorten < = 3pt, shorten >=3pt]

\tikzstyle{edge_php} = [edge, color=color_gadget_PHP!70!black]
\tikzstyle{edge_php_small} = [edge_php, edge_regular_small]

\tikzstyle{edge_eopl} = [edge, color=color_gadget_EOPL!70!black]
\tikzstyle{edge_eopl_small} = [edge_eopl, edge_regular_small]

\tikzstyle{edge_iter} = [edge, color=color_gadget_ITER!70!black]
\tikzstyle{edge_iter_small} = [edge_iter, edge_regular_small]

\tikzstyle{edge_pathphp_small} = [line width = 0.25mm, -{Latex[round]}, decorate, decoration={snake, segment length=2.5mm, amplitude=1mm, pre length=7pt,post length=8pt}, shorten < = 3pt, shorten >=3pt, color=color_gadget_PATHPHP_label]

\tikzstyle{edge_orange} = [edge, >=latex, <->, color=color_gadget_SOPL!90!black]

\tikzstyle{gadget} = [rounded corners, line width = 0.4mm, dashed]

\tikzstyle{gadget_PHP} = [gadget, draw = color_gadget_PHP, fill=color_gadget_PHP_inner]
\tikzstyle{gadget_PHP_small} = [gadget_PHP, line width = 0.2mm]

\tikzstyle{gadget_SOPL} = [gadget, draw = color_gadget_SOPL, fill=color_gadget_SOPL_inner]
\tikzstyle{gadget_SOPL_small} = [gadget_SOPL, line width = 0.2mm]

\tikzstyle{gadget_PATHPHP} = [gadget, draw = color_gadget_PATHPHP, fill=color_gadget_PATHPHP_inner]
\tikzstyle{gadget_PATHPHP_small} = [gadget_PATHPHP, line width = 0.2mm]

\tikzstyle{gadget_ITER} = [gadget, draw = color_gadget_ITER, fill=color_gadget_ITER_inner]
\tikzstyle{gadget_ITER_small} = [gadget_ITER, line width = 0.2mm]

\tikzstyle{gadget_EOPL} = [gadget, draw = color_gadget_EOPL, fill=color_gadget_EOPL_inner]
\tikzstyle{gadget_EOPL_small} = [gadget_EOPL, line width = 0.2mm]

\usepackage[font=small,labelfont=bf]{caption} 

\usepackage{lpic}

\makeatletter
\DeclareRobustCommand\bfseries{%
  \not@math@alphabet\bfseries\mathbf
  \fontseries\bfdefault\selectfont\boldmath}
\makeatother

\newcommand{\Xor}{\textsc{\upshape\scshape Xor}}

\newcommand{\Col}{\text{\upshape\scshape Col}}
\newcommand{\BiCol}{\text{\upshape\scshape BiCol}}
\newcommand{\PHP}{\text{\upshape\scshape PHP}}
\newcommand{\BiPHP}{\text{\upshape\scshape BiPHP}}
\newcommand{\VER}{\text{\upshape\scshape Ver}}
\newcommand{\Unfold}{\text{\upshape\scshape Unfold}}
\newcommand{\SetUnfold}{\text{\upshape\scshape SetUnfold}}

\newcommand{\calP}{\mathcal{P}}

\newcommand{\Z}{\mathbb{Z}}

\begin{document}
	
\maketitle

\begin{abstract}
The \emph{Collision problem} is to decide whether a given list of numbers $(x_1,\ldots,x_n)\in[n]^n$ is $1$-to-$1$ or $2$-to-$1$ when promised one of them is the case. We show an $n^{\Omega(1)}$ randomised communication lower bound for the natural two-party version of Collision where Alice holds the first half of the bits of each $x_i$ and Bob holds the second half. As an application, we also show a similar lower bound for a weak bit-pigeonhole search problem, which answers a question of Itsykson and Riazanov ({\footnotesize CCC 2021}).
\end{abstract}

\section{Introduction} \label{sec:results}

\paragraph*{Collision problem} The \emph{Collision problem} $\Col_N\colon[N]^N\to\{0,1,*\}$ is the following partial (promise) function. The input is a list of numbers~$z= (z_1,\ldots,z_N)\in [N]^N$ where $N$ is even. The goal is to distinguish between the following two cases, when promised that $z$ satisfies one of them.
\begin{itemize}
    \item $\Col_N(z)=0$ iff $z$ is $1$-to-$1$, that is, every number in the list $z$ appears in the list once.
    \item $\Col_N(z)=1$ iff $z$ is $2$-to-$1$, that is, every number in the list $z$ appears in the list twice.
\end{itemize}

The Collision problem has been studied exhaustively in quantum query complexity \cite{Brassard1998,Aaronson2002,Aaronson2004,GroverR04,Kutin2005,Ambainis2005,Aaronson2012,Scott2013,BunT16}. It was initially introduced to model the task of breaking collision resistant hash functions, a central problem in cryptanalysis. A robust variant of Collision is complete for {\sffamily NISZK}~\cite{Blum1991}, and consequently it has been featured in black-box oracle separations~\cite{Lovett2017,Bouland2019}. The problem has also been used in reductions to show hardness of other problems such as set-equality~\cite{Midrijanis2004} and various problems in property testing~\cite{Bravyi2011}. Upper bounds for Collision has been used to design quantum algorithms for triangle finding \cite{Magniez2007} and approximate counting \cite{AaronsonKKT20}.

In this paper, we consider a natural bipartite communication version of this problem, where we split the binary encoding of each number between two parties, Alice and Bob. Specifically, for~$N= 2^n$ where $n$ is even, we will define a bipartite function
\[
\BiCol_N\colon(\{0,1\}^{n/2})^N\times(\{0,1\}^{n/2})^N\to\{0,1,*\}.
\]
Here Alice gets as input a list of half-numbers $x=(x_1,\ldots,x_N)\in(\{0,1\}^{n/2})^N$, Bob gets a list of half-numbers $y=(y_1,\ldots,y_N)\in(\{0,1\}^{n/2})^N$, and we view their concatenation $z\coloneqq x\centerdot y$, defined by~$z_i\coloneqq x_iy_i$, as an input to $\Col_N$. Their goal is to compute $\BiCol_N(x,y)\coloneqq \Col_N(x\centerdot y)$.

\paragraph*{Upper bounds}
We first observe that $\BiCol_N$ admits a deterministic protocol that communicates at most $O(\sqrt{N}\log N)$ bits. Indeed, if  $x\centerdot y$ is 1--1, then since Alice's half-numbers are $n/2$ bits long, there are $\sqrt{N}$ distinct half-numbers, each appearing $\sqrt{N}$ many times in $x$. We may assume this is true also if $x\centerdot y$ is 2--1 (as otherwise it is easy to tell that we are in case 2--1). Consider the set of indices $I\coloneqq\{i\in[N]: x_i=0^{n/2}\}$, $|I|= \sqrt{N}$. Then $x\centerdot y$ restricted to indices~$I$ is~1--1 (resp.~2--1) if the original unrestricted input is 1--1 (resp.~2--1). Hence Alice can send the indices~$I$ to Bob, who can determine the value of the function.

If we are allowed randomness, we can do slightly better: there is a randomised protocol of cost~$O(N^{1/4}\log N)$. In this protocol, Alice samples a subset $I'\subseteq I$ of size $|I'|=\Theta(N^{1/4})$ uniformly at random and sends it to Bob, who checks for a collision in his part of the input. If the original input was 2--1, then by the birthday paradox, Bob will observe a collision with high probability.

\paragraph*{Lower bound}
As our main result, we prove a small polynomial lower bound for $\BiCol_N$, which shows that the above randomised protocol cannot be improved too dramatically.
\begin{theorem} \label{thm:main}
$\BiCol_N$ has randomised (and even quantum) communication complexity $\Omega(N^{1/12})$.
\end{theorem}
We conjecture that the $O(N^{1/4}\log N)$-bit protocol for $\BiCol_N$ is essentially optimal (up to logarithmic factors) for randomised protocols. It is an interesting open problem to close this gap.

\subsection{Application}

\paragraph*{Bit-pigeonhole principle}
We also show a lower bound for a search problem associated with the \emph{pigeonhole principle}. We define $\PHP^M_N$ where $M>N$ as the following search problem: On input $z=(z_1,\ldots,z_M)\in[N]^M$ the goal is to output a collision, that is, a pair of distinct indices $i,j\in[M]$ such that $z_i=z_j$. We note that $\PHP^M_N$ is a \emph{total} search problem (not a promise problem); it always has a solution since we require $M>N$. As before, we can turn $\PHP^M_N$ naturally into a bipartite communication search problem $\BiPHP^M_N$ where $N= 2^n$ so that
\begin{itemize}[noitemsep]
\item Alice holds $x=(x_1,\ldots,x_M)\in(\{0,1\}^{n/2})^M$;
\item Bob holds~$y=(y_1,\ldots,y_M)\in(\{0,1\}^{n/2})^M$; and
\item the goal is find a collision, that is, distinct $i,j\in[M]$ such that $x_iy_i=x_jy_j$.
\end{itemize}

\paragraph*{Lower bounds}
Itsykson and Riazanov~\cite{Itsykson2021} proved that $\BiPHP^{N+1}_N$ requires $\Omega(\sqrt{N})$ bits of randomised communication. Their proof was via a randomised reduction from set-disjoitness. A corollary of their result is that any proof system that can be efficiently simulated by randomised protocols (most notably, tree-like $\text{Res}(\oplus)$~\cite{Itsykson2020}) requires exponential size to refute bit-pigeonhole formulas featuring $N+1$ pigeons and $N$ holes. They asked whether a similar communication lower bound could be proved for the \emph{weak} pigeonhole principle with $M=2N$ pigeons and $N$ holes. We answer their question in the affirmative in the following theorem.
\begin{theorem} \label{thm:search}
$\BiPHP^{2N}_N$ has randomised (and even quantum) communication complexity $\Omega(N^{1/12})$.
\end{theorem}

Previously, Hrube{\v{s}} and Pudl{\'{a}}k~\cite{Hrubes2017} showed a small polynomial lower bound for $\BiPHP^M_N$ for every $M>N$ against deterministic (and even dag-like) protocols. By contrast, \cref{thm:search} is the first randomised lower bound in the $M=2N$ regime.

\subsection{Techniques}

Our proof of \cref{thm:main} proceeds as follows. A popular method to prove communication lower bounds is to start with a partial boolean function $f\colon\{0,1\}^n\to\{0,1,*\}$ that is hard to compute for decision trees and then apply a \emph{lifting theorem} (we use one due to Sherstov~\cite{Sherstov2011}) to conclude that the function $f\circ g$ obtained by composing $f$ with a small gadget $g\colon \Sigma\times\Sigma\to\{0,1\}$ is hard for communication protocols. Here $f\circ g\colon \Sigma^n\times\Sigma^n\to\{0,1,*\}$ is the communication problem where Alice holds $x\in\Sigma^n$, Bob holds $y\in\Sigma^n$, and their goal is to output
\[
(f\circ g)(x,y) ~\coloneqq~ f(g(x_1,y_1),\ldots,g(x_n,y_n)).
\]
A straightforward application of lifting often produces communication problems that are ``artificial'' since they are of the composed form. In particular, at first blush, it seems that the $\BiCol_N$ problem cannot be written in the form $f\circ g$ for any $f$ and any $g$ for which a lifting theorem holds. To address this issue, our main technical innovation is to show how the composed function $\Col_N\circ g$, where $g$ is a sufficiently ``regular'' gadget, can indeed be \emph{reduced} to the natural problem~$\BiCol_{N'}$. In this reduction, the input length will blow up polynomially, $N'= N^{\Theta(1)}$, which is the main reason why we only get a small polynomial lower bound. Our new reduction generalises a previous reduction from~\cite[\S6]{Itsykson2021}, which was tailored for the 2-bit $\Xor$ gadget.

To prove \cref{thm:search} we give a randomised \emph{decision-to-search} reduction from $\BiCol_N$ to~$\BiPHP^{2N}_N$. That is, we show that if there is an efficient randomised protocol for solving the \emph{total} search problem~$\BiPHP^{2N}_N$, then there is an efficient randomised protocol for solving the \emph{promise} problem~$\BiCol_N$. Given this reduction, \cref{thm:search} then follows from \cref{thm:main}. Similar style of randomised reductions have been considered in prior works~\cite{Raz1992,Huynh2012,Goos2018,Itsykson2021}, although they have always reduced from set-disjointness.

\section{Reductions and regular functions}

We assume some familiarity with communication complexity; see, e.g., the textbooks~\cite{Kushilevitz1997,Rao2020}.
In particular, it is often useful to view a bipartite function $f\colon \{0, 1\}^n \times \{0,1\}^n\to\{0,1\}$ as a~$2^n$-by-$2^n$ boolean matrix. We now give several definitions for the purposes of the proof of our main result.
\begin{definition}[Rectangular reduction]
For bipartite functions $f, g$ with domains $\{0, 1\}^n \times \{0,1\}^n$ and $\{0, 1\}^m \times \{0,1\}^m$, we write $f \leq g$ if there is a \emph{rectangular reduction} from $f$ to $g$, that is, there exist $a\colon \{0, 1\}^n \to \{0, 1\}^m$ and $b\colon \{0, 1\}^n \to \{0, 1\}^m$ such that $f(x, y) = g(a(x), b(y))$ for all $x,y$.
\end{definition}

Next, using basic language from group theory, we define a new class of highly symmetric boolean functions that we call \emph{regular}. (We borrow the term \emph{regular} from group theory where group actions satisfying the property in \cref{def:reg} below are called \emph{regular}.)

Let $\Pi_n$ denote the symmetric group on $[n]$, that is, the set of all permutations $[n]\to[n]$. Let~$S\subseteq \Pi_n\times \Pi_n$ be any group. We let $S$ act on the set $[n]\times[n]$ by permuting the rows and columns, that is, an element $s=(s^A,s^B)\in S$ acts on $(x,y)\in[n]\times[n]$ by $s\cdot (x,y)\coloneqq(s^A(x),s^B(y))$. For~$(x,y)\in[n]\times [n]$, we define its \emph{orbit} by $S\cdot(x,y)\coloneqq\{s\cdot (x,y):s\in S\}$.

\begin{definition}[Regular function] \label{def:reg}
A bipartite function $f\colon \{0, 1\}^{k} \times \{0, 1\}^{k} \to \{0, 1\}$ is \emph{regular} if there is a group $S\subseteq \Pi_{2^k}\times\Pi_{2^k}$ acting on the domain of $f$ such that the orbit of any $(x,y)\in f^{-1}(b)$, where $b\in\{0,1\}$, equals $f^{-1}(b)$, and, moreover, for every pair of inputs $(x_1, y_1), (x_2, y_2) \in f^{-1}(b)$ there is a unique~$s \in S$ such that $s\cdot(x_1,y_1)=s\cdot(x_2,y_2)$.
\end{definition}

It follows from the definition that $|S|=|f^{-1}(b)|=2^{2k-1}$ for both $b\in\{0,1\}$. A simple example of a regular function is the $2$-bit $\Xor$ function together with the 2-element group consisting of the identity map and the map $(x,y)\mapsto(\neg x,\neg y)$. However, the $\Xor$ function does not satisfy a fully general lifting theorem. This is why we consider the following more complicated gadget, called a \emph{versatile} gadget, which has been shown to satisfy various lifting theorems~\cite{Sherstov2011,Goos2018,Anshu2021}.
\begin{definition}
$\VER\colon \mathbb{Z}_4 \times \mathbb{Z}_4 \to \{0, 1\}$ is defined by $\VER(x, y) \coloneqq 1$ iff $x + y \pmod{4} \in \{2, 3\}$. 
\end{definition}

\begin{lemma}
$\VER$ is regular.
\end{lemma}
\begin{proof}
Consider the group $S\subseteq \Pi_4\times \Pi_4$ generated by the elements $(x,y)\mapsto(x+1,y-1)$ and~$(x,y)\mapsto (1-x,-y)$ where we use modulo 4 arithmetic. By explicit computations, we see that (here we list each element as a function of $(x,y)$)
\[
S ~=~
\Bigg\{
\begin{array}{llll}
 (x,y), & (x+1,y-1), &  (x+2,y-2), & (x+3,y-3),  \\
(1-x,-y), & (2-x,3-y), & (3-x,2-y), &(-x,1-y)
\end{array}
\Bigg\}.
\]
It is straightforward to check that $S$ gives rise to orbits $\VER^{-1}(0)$ and $\VER^{-1}(1)$; see \cref{fig:ver_gens}. Moreover, since $|S|=8=|\VER^{-1}(b)|$ for $b\in\{0,1\}$, the uniqueness property holds, too.
\end{proof}

Previously, \cite{Goos2018} showed that $\VER$ is \emph{random self-reducible}, that is, it admits a randomised reduction that maps any fixed input $(x,y)\in\VER^{-1}(b)$ into a uniform random input in $\VER^{-1}(b)$. It is easy to see that if a function is regular, then it is also random self-reducible (the random self-reduction is to apply a random symmetry from $S$). The converse, however, is unclear to us: If~$f$ is random self-reducible and balanced (meaning $|f^{-1}(0)|=|f^{-1}(1)|$), is it necessarily regular?

\begin{figure}
\centering
\begin{minipage}[b]{0.4\textwidth}
\centering
\include{figs/ver_matrix.tex}
\vspace{-7mm}
{\bf (a)}
\end{minipage}
\hspace{5mm}
\begin{minipage}[b]{0.4\textwidth}
\centering
\begin{tikzpicture}[scale=1.2]

\tikzstyle{node_regular} = [node_regular_intro, inner sep=4pt,draw=none,rounded corners=3pt]

\begin{scope}[xscale=2]
\node[node_regular] (P11) at (0,0) {$(1, 1)$};
\node[node_regular] (P21) at (0,1) {$(2, 0)$};
\node[node_regular] (P31) at (0,2) {$(3, 3)$};
\node[node_regular] (P41) at (0,3) {$(0, 2)$};
\node[node_regular] (P12) at (1,0) {$(0, 3)$};
\node[node_regular] (P22) at (1,1) {$(1, 2)$};
\node[node_regular] (P32) at (1,2) {$(2, 1)$};
\node[node_regular] (P42) at (1,3) {$(3, 0)$};
\end{scope}

\draw[edge_regular] (P11) -- (P21);
\draw[edge_regular] (P21) -- (P31);
\draw[edge_regular] (P31) -- (P41);
\draw[edge_regular] (P41) to [out=180, in=180] (P11);

\draw[edge_regular] (P12) -- (P22);
\draw[edge_regular] (P22) -- (P32);
\draw[edge_regular] (P32) -- (P42);
\draw[edge_regular] (P42) to [out=0, in=0] (P12);

\draw[edge_orange] (P11) -- (P12);
\draw[edge_orange] (P21) -- (P42);
\draw[edge_orange] (P31) -- (P32);
\draw[edge_orange] (P41) -- (P22);
\end{tikzpicture}
\vspace{-7mm}
{\bf (b)}
\end{minipage}
\caption{{\bf (a)} The bipartite function $\VER\colon\Z_4\times\Z_4\to\{0,1\}$. {\bf (b)} The group relative to which $\VER$ is regular is generated by two elements whose actions on $\VER^{-1}(1)$ are illustrated here. The first generator is $(x, y) \mapsto (x+1, y-1)$ (black arrows) and the second is $(x, y) \mapsto (1-x, -y)$ (orange arrows).}
\label{fig:ver_gens}
\end{figure}
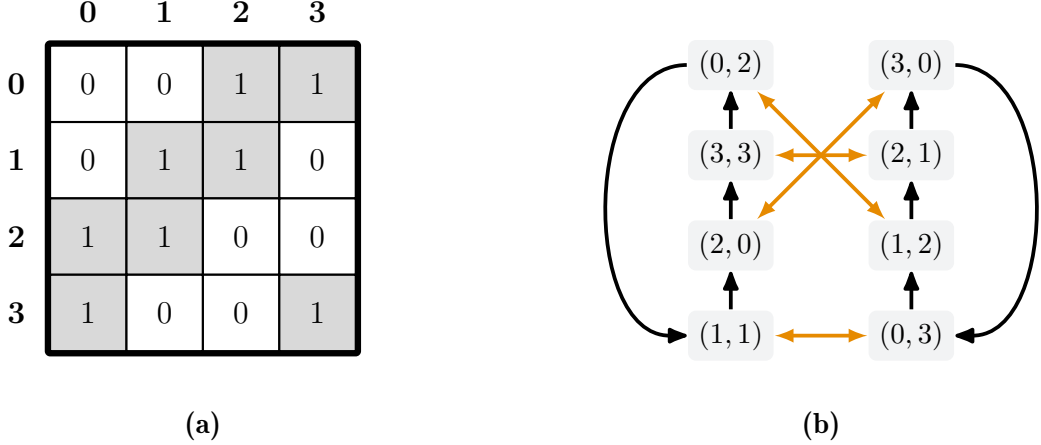

\section{Lower bound for bipartite collision}

In this section we prove \cref{thm:main}. We start with a standard application of a lifting theorem to establish a lower bound for the (somewhat artificial) composed function $\Col_N\circ\VER$. Here we think of $\Col_N$ as a boolean function $(\{0,1\}^n)^N\to\{0,1\}$ where $N=2^n$.
\begin{lemma} \label{lem:lift}
$\Col_N\circ \VER$ has randomised (and even quantum) communication complexity $\Omega(N^{1/3})$.
\end{lemma}
\begin{proof}
Aaronson and Shi~\cite{Aaronson2004} (building on~\cite{Aaronson2002}) showed that $\deg_{1/3}(\Col_N)\geq \Omega(N^{1/3})$ where~$\deg_{1/3}(f)$ for a partial boolean function $f$ is the least degree of a multivariate polynomial~$p(x)$ such that~$p(x)=f(x)\pm 1/3$ for all $x$ such that~$f(x)\in\{0,1\}$ and $|p(x)|\leq 4/3$ for all $x$ with~$f(x)=*$. Sherstov~\cite[\S12]{Sherstov2011} proved that for any partial boolean function $f$, we have that the randomised (and even quantum) communication complexity of $f\circ \VER$ is at least $\Omega(\deg_{1/3}(f))$. Combining these two results proves the lemma.
\end{proof}

The challenging part of the proof is to find a reduction from $\Col_N\circ g$ to $\BiCol_{N'}$ where $g$ is a regular gadget and $N'$ is polynomially larger than $N$. Choosing $g\coloneqq \VER$ in the following theorem and combining it with \cref{lem:lift} completes the proof of \cref{thm:main}. Note that the input length becomes~$N'\coloneqq N^4$ so that we obtain the lower bound $\Omega(N^{1/3})=\Omega(N'^{1/12})$, as claimed.
\begin{theorem}
\label{thm:reduction2dec}
Let $g\colon \{0, 1\}^{k} \times \{0, 1\}^{k} \to \{0, 1\}$ be a regular gadget. For every $N = 2^n$,
\[
\Col_N \circ g ~\leq~ \BiCol_{N^{2k}}.
\]
\end{theorem}

\begin{proof}
Consider the bipartite function $\Col_N\circ g$. Alice's input here is an $N$-tuple $(a^{(1)},\ldots,a^{(N)})$ where $a^{(j)}\in(\{0,1\}^k)^n$ for each $j\in[N]$. Bob's input $(b^{(1)},\ldots,b^{(N)})$ has a similar form. These bipartite inputs encode, via the gadgets, the input $(z^{(1)},\ldots,z^{(N)})$ to $\Col_N$ such that
\[
z^{(j)} ~\coloneqq~ g^n(a^{(j)},b^{(j)})
~\coloneqq~ (g(a^{(j)}_1,b^{(j)}_1),\ldots,g(a^{(j)}_n,b^{(j)}_n))
~\in~ \{0,1\}^n
\enspace \text{where}\enspace
a^{(j)}_i,b^{(j)}_i \in \{0,1\}^k.
\]
Let $S\subseteq\Pi_{2^k}\times\Pi_{2^k}$ be the symmetry group relative to which $g$ is regular. Recall that $|S|=2^{2k-1}$ and each $s\in S$ has the form $s=(s^A,s^B)$ with $s^A,s^B\in\Pi_{2^k}$. We fix an arbitrary ordering of the elements of $S$ and write $S(i)$ for the $i$-th element in this ordering. Thus $S=\{S(1),\ldots,S(2^{2k-1})\}$.

We first describe how the reduction expands each individual input $(a,b)\coloneqq (a^{(j)},b^{(j)})$ to $g^n$ into an ordered list of inputs to $g^n$. In more detail, the reduction
\begin{itemize}
\item takes an input $(a, b)=(a_1,\ldots,a_n,b_1,\ldots,b_n)\in(\{0,1\}^k)^{2n}$ to $g^n$, and
\item returns $\Unfold(a, b)\in(\{0,1\}^{2kn})^{N^{2k-1}}$, an \emph{ordered list} of $N^{2k-1}$ many inputs to $g^n$.
\end{itemize}
For any $n$-tuple of indices $I = (i_1, \ldots i_n)\in [|S|]^n$, we define the $I$-th pair in $\Unfold(a,b)$ by
\[
\Unfold(a, b)_I ~\coloneqq~ (\underbrace{s_{1}^{A}(a_1) s_{2}^{A}(a_2) \ldots s_{n}^{A}(a_{n})}_{\text{Alice's half}},~\underbrace{s_{1}^{B}(b_1) s_{2}^{B}(b_2) \ldots s_{n}^{B}(b_{n})}_{\text{Bob's half}})
\enspace \text{where}\enspace s_j \coloneqq S(i_j).
\]
Besides each pair in the list $\Unfold(a, b)$ being an input to $g^n$, we will also soon interpret them as pairs of half-numbers that are part of the input to $\BiCol_{N^{2k}}$.
Below, we write $\SetUnfold(a,b)\subseteq \{0,1\}^{2kn}$ for the \emph{set} of elements in the list $\Unfold(a,b)$, that is, ignoring the ordering and multiplicity of elements. 

\begin{claim} 
We have the following properties.
\begin{enumerate}[label=(\roman*),leftmargin=3em]
\item \label{lem:0}
$\SetUnfold(a, b)= (g^n)^{-1}(z)=g^{-1}(z_1)\times\cdots\times g^{-1}(z_n)$ where $z_i\coloneqq g(a_i,b_i)$.
\item \label{lem:1}
All pairs in $\Unfold(a, b)$ are distinct.
\item \label{lem:2}
Suppose $g^n(a, b) \neq g^n(a', b')$. Then $\SetUnfold(a, b) \cap \SetUnfold(a', b') = \emptyset$.
\item \label{lem:3}
Suppose $g^n(a, b) = g^n(a', b')$. Then $\SetUnfold(a, b)=\SetUnfold(a', b')$.
\end{enumerate}
\end{claim}
\begin{proof}
\Cref{lem:0}: Up to reordering of bits, the set equals $(S\cdot (a_1,b_1))\times (S\cdot (a_2,b_2))\times \cdots\times (S\cdot(a_n,b_n))$. By regularity, the orbit $S\cdot (a_i, b_i)$ is equal to $g^{-1}(z_i)$ for any $i$.
\Cref{lem:1}: The uniqueness property of the regular group action ensures that we do not get any repeated elements.
\Cref{lem:2}: If~$z\coloneqq g^n(a, b) \neq g^n(a', b')\eqqcolon z'$ then there is some $i$ such that $z_i\neq z'_i$. The $i$-th component of every pair in $\Unfold(a, b)$ lies in $g^{-1}(z_i)$ while the $i$-th component of every pair in $\Unfold(a, b)$ lies in $g^{-1}(z_i')$. The claim follows since these preimage sets are disjoint.
\Cref{lem:3}: If $g^{n}(a, b) = g^{n}(a', b')$, then \ref{lem:0} shows $\Unfold$ produces the same set for both $(a, b)$ and $(a', b')$.
\end{proof}

Our final reduction from $\Col_N\circ g$ maps Alice's $(a^{(1)},\ldots,a^{(N)})$ and Bob's $(b^{(1)},\ldots,b^{(N)})$ (which together encode the input $z=(z^{(1)},\ldots,z^{(N)})$ to $\Col_N$) to an input to $\BiCol_{N^{2k}}$ given by
\[
\Unfold(a^{(1)},b^{(1)}),\ldots,\Unfold(a^{(N)},b^{(N)}).
\]
Note that the reduction is rectangular: Alice can compute her part of the input, and Bob his.

It remains to check that the reduction treats 1--1 and 2--1 inputs correctly. If the input $z$ to $\Col_N$ is 1--1, then the reduction produces a 1--1 input by \ref{lem:1} and \ref{lem:2}. If the input $z$ to $\Col_N$ is 2--1 then for every index $i$ there is exactly one more index $j$ such that $z^{(i)}\coloneqq g^n(a^{(i)}, b^{(i)}) = g^n(a^{(j)}, b^{(j)})\eqqcolon z^{(j)}$. Hence, by \ref{lem:3} the lists $\Unfold(a^{(i)}, b^{(i)})$ and $\Unfold(a^{(j)}, b^{(j)})$ have every element colliding with each other. This produces a 2--1 input.
\end{proof}

\section{Lower bound for bipartite pigeonhole}
In this section we prove \cref{thm:search}. We do it by describing a reduction from the decision problem $\BiCol_N$ to the search problem $\BiPHP^{2N}_N$.

\begin{theorem}
\label{thm:reduction2search}
If there is a randomised protocol for $\BiPHP^{2N}_N$ of communication cost $d$, then there is a randomised protocol for $\BiCol_N$ of cost $O(d)$.
\end{theorem}
\begin{proof}
The proof idea is to start with an input to $\BiCol_N$ and then append it with more numbers to construct an input to $\BiPHP^{2N}_N$. Adding more numbers will create some new collisions in the input list, but our reduction will remember which collisions where ``planted'' during the reduction. We then randomly shuffle the input list so as to make the planted collisions indistinguishable from collisions (if any) coming from the original input to $\BiCol_N$. We now explain this in more detail.

Let $(x,y)$ be an input to $\BiCol_N$. That is, Alice holds $x=(x_1,\ldots,x_N)\in(\{0,1\}^{n/2})^N$ and Bob holds $y=(y_1,\ldots,y_N)\in(\{0,1\}^{n/2})^N$. In the reduction, we first append Alice's input by the \emph{planted} half-numbers $(a_1,\ldots,a_N)\in(\{0,1\}^{n/2})^N$ and Bob's input by the \emph{planted} half-numbers $(b_1,\ldots,b_N)\in(\{0,1\}^{n/2})^N$ where the concatenated strings $a_ib_i$, $i\in[N]$, range lexicographically over all binary numbers in $\{0, 1\}^n$.

Next, Alice and Bob use public randomness to sample a permutation $\pi\colon[2N]\to[2N]$ uniformly at random, which they then use to permute their lists of length $2N$. While doing so, they remember which positions in the permuted list occupy planted numbers (namely, those in positions $\pi(\{N+1,\ldots,2N\})$). Call the resulting list~$(x',y')$. We now let Alice and Bob run the hypothesised protocol $\calP$ for $\BiPHP^{2N}_N$ on input $(x',y')$ to find some collision $x'_iy'_i=x'_jy'_j$ where $i\neq j$. (We assume for simplicity that $\calP$ finds a collision with probability $1$. The following analysis can be adapted even when $\calP$ errs with bounded probability.)

We have two cases depending on whether $(x,y)$ was 1--1 or 2--1 (see~\cref{fig:inputs}):
\begin{itemize}
    \item If $(x,y)$ was 1--1 then $(x',y')$ is 2--1. Moreover, each collision in $(x',y')$ involves a planted number. In particular, the collision $\{i,j\}$ found by the protocol always features at least one planted number.
    \item If $(x,y)$ was 2--1 then $(x',y')$ is an input where $N/2$ many numbers appear thrice, and $N/2$ numbers appear once. We claim that the collision $\{i,j\}$ found by $\calP$ will not feature a planted number with probability at least $1/3$ (over the random choice of $\pi$). Indeed, let~$k\notin\{i,j\}$ be the third position such that $x'_iy'_i=x'_jy'_j=x'_ky'_k$. Then conditioned on $\pi$ having produced the input $(x',y')$, each position in $\{i,j,k\}$ is equally likely to occupy a planted number. Thus, with probability $1/3$, the planted number lies in position $k$ and not in $\{i,j\}$.
\end{itemize}

Our protocol for $\BiCol_N$ guesses that $(x,y)$ is 2--1 if the collision $\{i,j\}$ returned by $\calP$ does not involve a planted number. We can further reduce the error probability down to~$(2/3)^t$ by repeating the randomised reduction and $\calP$ some $t=O(1)$ times and seeing if any one of these runs finds a collision without a planted number.
\end{proof}

\begin{figure}[t]
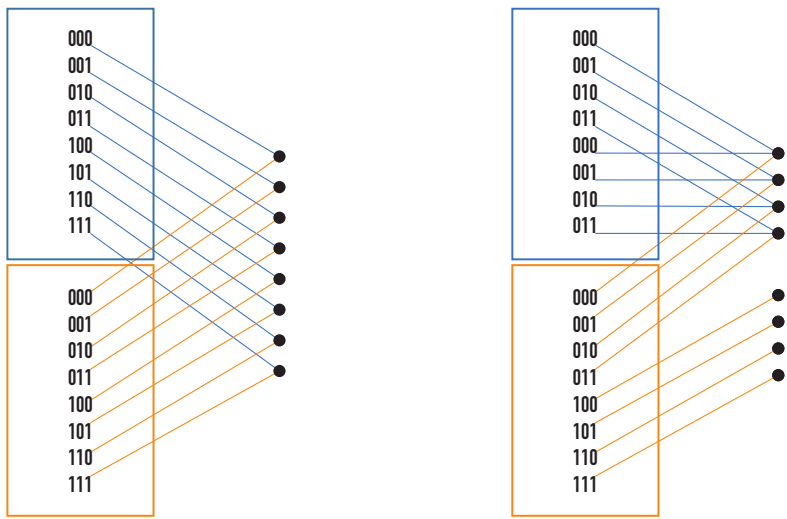

\centering
\begin{lpic}[t(-6mm)]{figs/inputs(0.4)}
\end{lpic}
\caption{Illustration of collisions in 1--1 and 2--1 inputs. The original input $(x,y)$ is drawn at the top, and the planted numbers $(a,b)$ are drawn at the bottom.}
\label{fig:inputs}
\vspace{-1mm}
\end{figure}


\bibliography{collision}

\newcommand{\etalchar}[1]{$^{#1}$}
\begin{thebibliography}{AKKT20}

\bibitem[Aar02]{Aaronson2002}
Scott Aaronson.
\newblock Quantum lower bound for the collision problem.
\newblock In {\em Proceedings of the 34th Symposium on Theory of Computing
  (STOC)}, pages 635--642. {ACM}, 2002.
\newblock \href {https://doi.org/10.1145/509907.509999}
  {\path{doi:10.1145/509907.509999}}.

\bibitem[Aar12]{Aaronson2012}
Scott Aaronson.
\newblock Impossibility of succinct quantum proofs for collision-freeness.
\newblock {\em Quantum Information and Computation}, 12(1-2):21--28, 2012.
\newblock \href {https://doi.org/10.26421/QIC12.1-2-3}
  {\path{doi:10.26421/QIC12.1-2-3}}.

\bibitem[Aar13]{Scott2013}
Scott Aaronson.
\newblock The collision lower bound after 12 years, 2013.
\newblock QStart talk.
\newblock URL: \url{https://scottaaronson.blog/?p=1458}.

\bibitem[ABK21]{Anshu2021}
Anurag Anshu, Shalev Ben{-}David, and Srijita Kundu.
\newblock On query-to-communication lifting for adversary bounds.
\newblock In {\em Proceedings of the 36th Computational Complexity Conference
  (CCC)}, volume 200, pages 30:1--30:39. Schloss Dagstuhl, 2021.
\newblock \href {https://doi.org/10.4230/LIPICS.CCC.2021.30}
  {\path{doi:10.4230/LIPICS.CCC.2021.30}}.

\bibitem[AKKT20]{AaronsonKKT20}
Scott Aaronson, Robin Kothari, William Kretschmer, and Justin Thaler.
\newblock Quantum lower bounds for approximate counting via laurent
  polynomials.
\newblock In Shubhangi Saraf, editor, {\em 35th Computational Complexity
  Conference, {CCC} 2020, July 28-31, 2020, Saarbr{\"{u}}cken, Germany (Virtual
  Conference)}, volume 169 of {\em LIPIcs}, pages 7:1--7:47. Schloss Dagstuhl -
  Leibniz-Zentrum f{\"{u}}r Informatik, 2020.
\newblock \href {https://doi.org/10.4230/LIPIcs.CCC.2020.7}
  {\path{doi:10.4230/LIPIcs.CCC.2020.7}}.

\bibitem[Amb05]{Ambainis2005}
Andris Ambainis.
\newblock Polynomial degree and lower bounds in quantum complexity: collision
  and element distinctness with small range.
\newblock {\em Theory Comput.}, 1:37--46, 2005.
\newblock \href {https://doi.org/10.4086/toc.2005.v001a003}
  {\path{doi:10.4086/toc.2005.v001a003}}.

\bibitem[AS04]{Aaronson2004}
Scott Aaronson and Yaoyun Shi.
\newblock Quantum lower bounds for the collision and the element distinctness
  problems.
\newblock {\em Journal of the {ACM}}, 51(4):595--605, jul 2004.
\newblock \href {https://doi.org/10.1145/1008731.1008735}
  {\path{doi:10.1145/1008731.1008735}}.

\bibitem[BCH{\etalchar{+}}19]{Bouland2019}
Adam Bouland, Lijie Chen, Dhiraj Holden, Justin Thaler, and Prashant~Nalini
  Vasudevan.
\newblock On the power of statistical zero knowledge.
\newblock {\em {SIAM} Journal on Computing}, 49(4):FOCS17--1--FOCS17--58, 2019.
\newblock \href {https://doi.org/10.1137/17m1161749}
  {\path{doi:10.1137/17m1161749}}.

\bibitem[BHH11]{Bravyi2011}
Sergey Bravyi, Aram Harrow, and Avinatan Hassidim.
\newblock Quantum algorithms for testing properties of distributions.
\newblock {\em IEEE Transactions on Information Theory}, 57(6):3971--3981,
  2011.
\newblock \href {https://doi.org/10.1109/TIT.2011.2134250}
  {\path{doi:10.1109/TIT.2011.2134250}}.

\bibitem[BHT98]{Brassard1998}
Gilles Brassard, Peter H{\o}yer, and Alain Tapp.
\newblock Quantum cryptanalysis of hash and claw-free functions.
\newblock In {\em Proceedings of the 3rd Latin American Symposium on
  Theoretical Informatics (LATIN)}, pages 163--169. Springer, 1998.

\bibitem[BSMP91]{Blum1991}
Manuel Blum, Alfredo~De Santis, Silvio Micali, and Giuseppe Persiano.
\newblock Noninteractive zero-knowledge.
\newblock {\em {SIAM} Journal on Computing}, 20(6):1084--1118, 1991.
\newblock \href {https://doi.org/10.1137/0220068} {\path{doi:10.1137/0220068}}.

\bibitem[BT16]{BunT16}
Mark Bun and Justin Thaler.
\newblock Dual polynomials for collision and element distinctness.
\newblock {\em Theory Comput.}, 12(1):1--34, 2016.
\newblock \href {https://doi.org/10.4086/toc.2016.v012a016}
  {\path{doi:10.4086/toc.2016.v012a016}}.

\bibitem[GP18]{Goos2018}
Mika G{\"o}{\"o}s and Toniann Pitassi.
\newblock Communication lower bounds via critical block sensitivity.
\newblock {\em SIAM Journal on Computing}, 47(5):1778--1806, 2018.
\newblock \href {https://doi.org/10.1137/16M1082007}
  {\path{doi:10.1137/16M1082007}}.

\bibitem[GR04]{GroverR04}
Lov~K. Grover and Terry Rudolph.
\newblock How significant are the known collision and element distinctness
  quantum algorithms?
\newblock {\em Quantum Inf. Comput.}, 4(3):201--206, 2004.
\newblock \href {https://doi.org/10.26421/QIC4.3-5}
  {\path{doi:10.26421/QIC4.3-5}}.

\bibitem[HN12]{Huynh2012}
Trinh Huynh and Jakob Nordstr{\"o}m.
\newblock On the virtue of succinct proofs: Amplifying communication complexity
  hardness to time--space trade-offs in proof complexity.
\newblock In {\em Proceedings of the 44th Symposium on Theory of Computing
  (STOC)}, pages 233--248. ACM, 2012.
\newblock \href {https://doi.org/10.1145/2213977.2214000}
  {\path{doi:10.1145/2213977.2214000}}.

\bibitem[HP17]{Hrubes2017}
Pavel Hrube{\v{s}} and Pavel Pudl{\'{a}}k.
\newblock Random formulas, monotone circuits, and interpolation.
\newblock In {\em Proceedings of the 58th Symposium on Foundations of Computer
  Science (FOCS)}, pages 121--131, 2017.
\newblock \href {https://doi.org/10.1109/FOCS.2017.20}
  {\path{doi:10.1109/FOCS.2017.20}}.

\bibitem[IR21]{Itsykson2021}
Dmitry Itsykson and Artur Riazanov.
\newblock Proof complexity of natural formulas via communication arguments.
\newblock In {\em Proceedings of 36th Computational Complexity Conference
  (CCC)}, volume 200, pages 3:1--3:34. Schloss Dagstuhl, 2021.
\newblock \href {https://doi.org/10.4230/LIPIcs.CCC.2021.3}
  {\path{doi:10.4230/LIPIcs.CCC.2021.3}}.

\bibitem[IS20]{Itsykson2020}
Dmitry Itsykson and Dmitry Sokolov.
\newblock Resolution over linear equations modulo two.
\newblock {\em Annals of Pure and Applied Logic}, 171(1):1--31, 2020.
\newblock \href {https://doi.org/10.1016/j.apal.2019.102722}
  {\path{doi:10.1016/j.apal.2019.102722}}.

\bibitem[KN97]{Kushilevitz1997}
Eyal Kushilevitz and Noam Nisan.
\newblock {\em Communication Complexity}.
\newblock Cambridge University Press, 1997.
\newblock \href {https://doi.org/10.1017/CBO9780511574948}
  {\path{doi:10.1017/CBO9780511574948}}.

\bibitem[Kut05]{Kutin2005}
Samuel Kutin.
\newblock Quantum lower bound for the collision problem with small range.
\newblock {\em Theory of Computing}, 1(2):29--36, 2005.
\newblock \href {https://doi.org/10.4086/toc.2005.v001a002}
  {\path{doi:10.4086/toc.2005.v001a002}}.

\bibitem[LZ17]{Lovett2017}
Shachar Lovett and Jiapeng Zhang.
\newblock On the impossibility of entropy reversal, and its application to
  zero-knowledge proofs.
\newblock In {\em Proceedings of the 15th Theory of Cryptography Conference
  (TCC)}, pages 31--55. Springer, 2017.
\newblock \href {https://doi.org/10.1007/978-3-319-70500-2_2}
  {\path{doi:10.1007/978-3-319-70500-2_2}}.

\bibitem[Mid04]{Midrijanis2004}
Gatis Midrij{\={a}}nis.
\newblock A polynomial quantum query lower bound for the set equality problem.
\newblock In {\em Proceedings of the 31st International Conference on Automata,
  Languages and Programming (ICALP)}, volume 3142, pages 996--1005. Springer,
  2004.
\newblock \href {https://doi.org/10.1007/978-3-540-27836-8_83}
  {\path{doi:10.1007/978-3-540-27836-8_83}}.

\bibitem[MSS07]{Magniez2007}
Fr{\'{e}}d{\'{e}}ric Magniez, Miklos Santha, and Mario Szegedy.
\newblock Quantum algorithms for the triangle problem.
\newblock {\em {SIAM} Journal on Computing}, 37(2):413--424, January 2007.
\newblock \href {https://doi.org/10.1137/050643684}
  {\path{doi:10.1137/050643684}}.

\bibitem[RW92]{Raz1992}
Ran Raz and Avi Wigderson.
\newblock Monotone circuits for matching require linear depth.
\newblock {\em Journal of the ACM}, 39(3):736–744, jul 1992.
\newblock \href {https://doi.org/10.1145/146637.146684}
  {\path{doi:10.1145/146637.146684}}.

\bibitem[RY20]{Rao2020}
Anup Rao and Amir Yehudayoff.
\newblock {\em Communication Complexity: And Applications}.
\newblock Cambridge University Press, 2020.
\newblock \href {https://doi.org/10.1017/9781108671644}
  {\path{doi:10.1017/9781108671644}}.

\bibitem[She11]{Sherstov2011}
Alexander Sherstov.
\newblock The pattern matrix method.
\newblock {\em {SIAM} Journal on Computing}, 40(6):1969--2000, 2011.
\newblock \href {https://doi.org/10.1137/080733644}
  {\path{doi:10.1137/080733644}}.

\end{thebibliography}

\end{document}